%% file: svdcoreset_arxiv.tex
\documentclass{article}

\usepackage{graphicx}
\usepackage[caption=false]{subfig}
\usepackage{algorithm}
\usepackage{algorithmic}
\usepackage[algo2e]{algorithm2e}
\usepackage{amsmath, amsthm}
\usepackage{url}
\usepackage{fullpage}
\usepackage{amsfonts}
\usepackage{amssymb}
\usepackage{subfig}

\include{definitions}

\begin{document}
\date{}
\title{Dimensionality Reduction of Massive Sparse Datasets Using Coresets}
\author{Dan Feldman, Mikhail Volkov, and Daniela Rus}
\maketitle

%%%%%%%%%%%%%%%%%%%%%%%%%%%%%%%%%%%%%%%%%%%%%%%%%%%%%%%%%%%%%%%%%%%%%%%%%%%%%%%%%%%%%%%%%%%%%%%%%%%%%%%%%

\input{abstract}
\input{introduction}
\input{approach}
\input{construction}
\input{meancoresets}

\input{hyperplane}
\input{krank}
\input{itemfreq}

\input{results}
\input{conclusion}

%%%%%%%%%%%%%%%%%%%%%%%%%%%%%%%%%%%%%%%%%%%%%%%%%%%%%%%%%%%%%%%%%%%%%%%%%%%%%%%%%%%%%%%%%%%%%%%%%%%%%%%%%

\bibliographystyle{plain}
\bibliography{references}

\end{document}

%% file: definitions.tex
\providecommand{\norm}[1]{\left\lVert#1\right\rVert}

\newcommand{\br}[1]{\left\{#1\right\}}
\newcommand{\REAL}{\ensuremath{\mathbb{R}}}
\newcommand{\eps}{\varepsilon}
\newtheorem{theorem}{Theorem}

\newtheorem{definition}{Definition}
\newtheorem{claim}{Claim}
\newtheorem{corollary}{Corollary}
\newcommand{\dist}{\mathrm{dist}}

\newcommand{\SVDCoreset}{\textsc{SVD-Coreset}}

\newcommand{\set}[1]        {\left\{ #1 \right\}}

\newcommand{\rbr}[1]      {\left( #1 \right)}

\def\norm#1{\|#1\|}

%% file: abstract.tex
\begin{abstract}
In this paper we present a practical solution with performance guarantees to the problem of dimensionality reduction for very large scale sparse matrices.  We show applications of our approach to computing the low rank approximation (reduced SVD) of such matrices. Our solution uses coresets, which is a subset of $O(k/\eps^2)$ scaled rows from the $n\times d$ input matrix, that approximates the sub of squared distances from its rows to every $k$-dimensional subspace in $\REAL^d$, up to a factor of $1\pm\eps$. An open theoretical problem has been whether we can compute such a coreset that is independent of the input matrix and also a weighted subset of its rows. 
%An open practical problem has been whether we can compute a non-trivial approximation to the reduced SVD of very large databases such as the Wikipedia document-term matrix in a reasonable time. 
We answer this question affirmatively. % and demonstrate an algorithm that efficiently computes a low rank approximation of the entire English Wikipedia.
Our main technical result is a novel technique for deterministic coreset construction that is based on a reduction to the problem of $\ell_2$ approximation for item frequencies. 
%We demonstrate it by constructing the first coresets of size independent of input size, for two problems: sum of squared distances to every query point and $k$-dimensional subspace.
\end{abstract} 

%% file: introduction.tex
%==============================================================================================================
\section{Introduction}
\label{sec:introduction}

Algorithms for dimensionality reduction usually aim to project an
input set of $d$-dimensional vectors (database records) onto a $k\leq
d-1$ dimensional affine subspace that minimizes the sum of squared
distances to these vectors, under some constraints. Special cases
include Principle Component Analysis (PCA), Linear regression
($k=d-1$), Low-rank approximation ($k$-SVD), Latent Drichlet Analysis
(LDA) and Non-negative Matrix Factorization (NNMF). Learning
algorithms such as $k$-means clustering can then be applied on the
low-dimensional data to obtain fast approximations with provable
guarantees.

However, there are still no practical algorithms with
provable guarantees that compute dimension reduction for sparse
large-scale data. Much of the large scale high-dimensional data sets
available today (e.g. image streams, adjacency matrices of graphs and social networks, text streams, etc.) are
sparse. For example, consider the text case of the English Wikipedia. We can
associate a matrix with the Wikipedia, where the (usually English) words define
the columns (approximately 8 millions terms) and the individual documents
define the rows (approximately 3.7 million documents).
This large scale matrix is sparse because every document uses a very small fraction of the possible English words.
 %
%Due to the computation limits of current SVD implementations, research as in~\cite{wikisurvey} focuses on ad-hoc filtering rules to choose the best small subset of words and documents from Wikipedia for the SVD stage. E.g., frequency and length of words and documents, or formatting such as only hyper-linked words or titles. This practical open problem matches the corresponding theoretical open problem of computing reduced SVD for large matrices that is considered in this paper. Our approach can be considered as an automatic way with guarantees to solve this data reduction problem, instead of using ad-hoc rules made by humans.

In this paper we present an algorithm for dimensionality reduction of very large scale sparse data sets such as the Wikipedia with provable approximations.  Our approach uses coresets to solve the time and space challanges.

Given a matrix $A$, a coreset $C$ in this paper is defined as a
weighted subset of rows of $A$ such that the sum of squared distances from any given $k$-dimensional subspace to the rows of $A$ is approximately the same as the sum of squared weighted distances to the rows in $C$. Formally,
\begin{definition}[$(\eps,k)$-Coreset]\label{def}
Given a $n\times d$ matrix $A$ whose rows $a_1,\cdots,a_n$ are $n$ points (vectors) in $\REAL^d$, an error parameter $\eps\in(0,1]$, and an integer $k\in[1,d-1]$ that represents the desired dimensionality reduction, an $(\eps,k)$-coreset for $A$ is a weighted subset $C=\br{w_i a_i \mid w_i>0 , i\in\br{1,\cdots,n}}$ of the rows of $A$, where $w=(w_1,\cdots,w_n)\in[0,\infty)^n$ is a non-negative weight vector such that for every $k$-subspace $S$ in $\REAL^d$ we have
\[
\begin{split}
&\left|\sum_i (\dist(a_i,S))^2-\sum_i (\dist(w_i a_i,S))^2\right|\\
&\leq \eps \sum_i (\dist(a_i,S))^2.
\end{split}
\]
\end{definition}

That is, the sum of squared distances from the $n$ points to $S$ approximates the sum of squared weighted distances
$\sum_{i=1}^n w^2_i(\dist(a_i,S))^2$ to $S$. The approximation is up to a multiplicative factor of $1\pm\eps$.
By choosing $w=(1,\cdots,1)$ we obtain a trivial $(0,k)$-coreset. However, in a more efficient coreset most of the weights will be zero and the corresponding rows in $A$ can be discarded. The cardinality of the coreset is thus the sparsity of $w$
\[
|C|=\norm{w}_0:=|\br{w_i\neq 0\mid i\in\br{1,\cdots,n}}|.%=\frac{k^2}{\eps^2}
\]

If $C$ is small, then the computation is efficient. Because $C$
is a weighted subset of the rows of $A$, if $A$ is sparse, then $C$ is also sparse. A long-open
research question has been whether we can have such a coreset that is both \emph{of size} independent of the input dimension ($n$ and $d$) and a
\emph{subset of the original input rows}. In this paper we answer this
question affirmatively as follows.
\begin{theorem}\label{thm:one}
For every input matrix $A\in\REAL^{n\times d}$, a parameter error $\eps\in(0,1]$ and an integer $k\geq1$, there is a  $(k,\eps)$-coreset of size $|C|=O(k/\eps^2)$; see Definition~\eqref{def}.
\end{theorem}

Our proof is constructive and we provide an efficient algorithm for computing such a coreset $C$.
%In our experimental results we show that this algorithm provides a practical solution to a
%long-standing open practical problem: computing the reduced SVD of large
%matrices such as those associated with Wikipedia.

%--------------------------------------------------------------------------------------------------------------
\subsection{Why Coresets?}
\label{sec:coreset}
The motivation for using data reduction technique, and coresets as defined above in particular, is given below.

\textbf{Fast approximations.}
The $k$-subspace that minimizes the sum of squared distances to the vectors in $C$, will minimize the sum of squared distances to the rows of $A$, up to a factor of $(1\pm\eps)$, over every $k$-subspace in $\REAL^d$.
If the $\eps$-coreset $C$ is small, it can used to efficiently compute the low-rank approximation (reduced SVD) of the original matrix. More generally, using an $\eps$-coreset we can compute the optimal $k$-subspace $S$ under some given constraints, as desired by several popular dimensionality reduction techniques.

\textbf{Streaming and parallel computation.} An algorithm for constructing an $\eps$-coreset off-line on a single machine, can be easily turned into an algorithm that maintains an $\eps$-coreset $C$ for a (possibly unbounded) stream of $d$-dimensional input vectors, $a_1,\cdots,a_n$  where $n$ is the number of vectors seen so far, using one pass over these vectors and $O(|C|\log n)$ memory. 
Using the same merge-and-reduce technique, we can compute $C$ in a way known as embarrassingly parallel on distributed $M$ machines on a cloud, network or GPGPU, and reduce the construction time of $C$ by a factor of $M$. In both the streaming and parallel models, the size of the coreset $C$ will be increased where $\eps$ is replaced by $O(\eps/\log n)$ in the memory and running time computations. For our specific technique of the Frank-Wolfe algorithm, similar streaming versions do not introduce the $\log n$ factor.

\textbf{Applications to sparse data.}
Let $s$ denote the sparsity, i.e., maximum non-zero entries, in a row of $A$, over all of its $n$ rows. Note that this definition is different than the more common definition of number of non-zeroes entries (nnz) in the entire matrix. The memory (space) that is used by $C$ is then $O(|C|\cdot m)$ words (real numbers). In particular, if the cardinality of $C$ is independent of the dimensions $n$ and $d$ of $A$, then so is the required memory to store $C$ and to compute its low rank approximation.
This property allows us to handle massive datasets of unbounded dimensions and rows, with the natural limitation on the sparsity of each row.

\textbf{Sparse reduced SVD.} A lot of algorithms and papers aim to compute a low rank approximation which can be spanned by few input points.  In particular, if each input point (row) is sparse then this low rank approximation will be sparse.
For example, in the case of text mining, each topic (singular vector) will be a linear combination of few words, and not all the words in the dictionary.
We get this property ``for free" in the case of coresets: if the coreset consists of few input rows, its optimal $k$-subspace ($k$-rank approximation) will also be spanned by few input points. The optimal $k$-subspace of the sketches presented in Section \ref{sec:related} do not have this property in general.

\textbf{Interpretation and Factorization.} Since coreset consists of few weighted rows, it tells us which records in the data are most importance, and what is their importance. For example, a coreset for the Wikipedia tells which documents are most important for approximating the reduced SVD (Latent Semantic Analysis).
This and the previous property are the motivations behind the well known $CX$ and $CUR$ decompositions, where $C$ is a subset of columns from the transpose input matrix $A^T$, $R$ is a subset of rows from $A^T$. A coreset represents an even simpler $WA$ decomposition, where $W$ is a diagonal sparse matrix that consists of the weights $w_1,\cdots,w_n$ in its diagonal, or a  $DR$ decomposition where $D$ is a diagonal matrix that consists of the non-zero rows of $W$, and $R$ is a subset of rows from $A$.

%--------------------------------------------------------------------------------------------------------------
\subsection{Our contribution}
Our main result is the first algorithm for computing an $(\eps,k)$-coreset $C$ of size independent of both $n$ and $d$, for any given $n\times d$ input matrix .
%It can be computed efficiently for every $\eps\in(0,1]$ and integer $k\geq 1$.
Until today, it was not clear that such a coreset exists. The polynomial dependency on $d$ of previous coresets made them useless for fat or square input matrices, such as Wikipedia, images in a sparse feature space representation, or adjacency matrix of a graph.
The output coreset has all the properties that are discussed in Section~\ref{sec:coreset}.
 %with the additional advantage that it can handle unbounded input stream of vectors of unbounded dimensions, as long as the sparsity (number of non-zeroes) of each input row is bounded by $s$.
We summarize this result in Theorem~\ref{thm:one}.

% and its construction is deterministic
\textbf{Open problems.} We answer in this paper two recent open problems: coresets of size independent of $d$ for low rank approximation~\cite{ghashami2015frequent} and for $1$-mean queries~\cite{clarkson2010coresets}.
Both of the answers are based on our new technical tool for coreset construction; see Section~\ref{novel}.
%Such sparse vectors are usually represented as in Matlab by $s$ pairs fo the form $(i,val)$ where $i$ is the index of a non-zero entry, and $val$ is its value.
%More precisely, we provide an algorithm that takes as input a set $A$ of $n$ vectors $a_1,\cdots,a_n$ in $\REAL^d$, an
%approximation error $\eps\in(0,1]$, and an integer $k\in[1,d-1]$.
%It then returns an $\eps$-coreset $C=\br{w_i a_i}_{i=1}^n$ for $A$  of size $|C|=\norm{w}_0=O(k^2/\epsilon^2)$.

\textbf{Efficient construction.} We suggest a deterministic algorithm that efficiently compute the above coreset for every given matrix. The construction time is dominated by the time it takes to compute the $k$-rank approximation for the input matrix. 
Any such $(\eps,k)$-coreset can be computed for both the parallel and streaming models as explained in~\cite{FSS13}, or using existing streaming versions for the Frank-Wolfe algorithm.

\textbf{Efficient Implementation.}
Although our algorithm runs on points in $R^{d^2}$ that represent $d\times d$ matrices, but we show how it can actually be implemented using operations in $\REAL^d$, similarly to learning kernels techniques. The running time of the algorithm is dominated by the time it takes to compute a constant factor approximation to the $k$-rank approximation of the input matrix, and the time it takes to approximates the farthest input point from a given query point in the input. Over the recent years several efficient algorithms were suggested for these problems using pre-processing, in particular for sparse data, (e.g.~\cite{clarkson2013low}, Fast JL and farthest nearest neighbour).
%
%Using recent algorithms for computing reduced SVD for (small, non-streaming) sparse data (e.g.~\cite{clarkson2013low}), and approximation algorithms for the farthest points~\cite{EldLinPorZee97}, the running time can be significantly improved. In addition, our experimental results imply that our analysis is not tight and the actual coreset size that is returned by our algorithm may be $O(k/\eps^2)$, i.e., linear in $k$ and meets the known lower bound. This may be obtained by an improved and recent analysis of the Frank-Wolfe algorithm in~\cite{clarkson2010coresets}.
%The coreset that we suggessize $O(k^2/\eps^2)$, however, in the experiments we see results that are linear in $k$. We suspect that a better analysis of our algorithm (e.g. the recent results in~\cite{} for the Frank-Wolfe algorithm) will prove

%--------------------------------------------------------------------------------------------------------------
\subsection{Related Work}
\label{sec:related}

\textbf{Existing Software and implementations.}
We have tried to run modern reduced SVD implementations such as GenSim~\cite{gensym} that uses random projections or Matlab's \emph{svds} function that uses the classic LAPACK library that in turn uses a variant of the Power (Lanczoz) method for handling sparse data. All of them crashed for an input of few thousand of documents and $k<100$, as expected from the analysis of their algorithms. 
%See Section~\ref{sec:experiments}. 
Even for $k=3$, running the implementation of svds in Hadoop~\cite{shvachko2010hadoop} (also a version of the Power Method) took several days~\cite{halko2012randomized}.

%\label{sec:related}
\textbf{Coresets.}
Following a decade of research, in~\cite{varadarajan2012sensitivity} it was recently proved that an $(\eps,k)$ coreset of size $|C|=O(dk^3/\eps^2)$ exists for every input matrix. The proof is based on a general framework for constructing different kinds of coresets, and is known as \emph{sensitivity}~\cite{FL10, LS10}.
This coreset is efficient for tall matrices, since its cardinality is independent of $n$. However, it is useless for ``fat" or square matrices (such as the Wikipedia matrix above), where $d$ is in the order of $n$, which is the main motivation for our paper. In~\cite{clarkson2010coresets}, the Frank-Wolfe algorithm was used to construct different types of coresets than ours, and for different problems. Our approach is based on a solution that we give to an open problem in~\cite{clarkson2010coresets}, however we can see how it can be used to compute the coresets in~\cite{clarkson2010coresets} and vice versa.

\textbf{Sketches.} %A sketch is sometimes used to denote compressed version of the data that is based on some linear combinations of the input data, such as in compressed sensing~\cite{}.
 %Sketches are based on multiplying the $n\times d$ input matrix $A$ (whose $i$th row is $a_i$) from the left by a fat ``sketch matrix" $S$ that has only few rows. The resulting matrix $SA$ thus has few rows, where each row is a linear combination of some or all the input rows. Geometrically, this is similar to projecting the input points (rows) on some low-dimensional subspace. For example, the classic Johnson–Lindenstrauss lemma state that projecting points on a random subspace nearly preserves the distances between pair of points. The special case where $S$ has a single non-zero entry in each row yields a strong coreset (weighted subset of the input) as defined above.
A \emph{sketch} in the context of matrices is a set of vectors $u_1,\cdots,u_s$ in $\REAL^d$ such that the sum of squared distances $\sum_{i=1}^n (\dist(a_i,S))^2$ from the input $n$ points to \emph{every} $k$-dimensional subspace $S$ in $\REAL^d$, can be approximated by $\sum_{i=1}^n (\dist(u_i,S))^2$ up to a multiplicative factor of $1\pm\eps$.

Note that even if the input vectors $a_1,\cdots,a_n$ are sparse, the sketched vectors $u_1,\cdots,u_s$ in general are not sparse, unlike the case of coresets.
A sketch of cardinality $d$ can be constructed with no approximation error ($\eps=0$), by defining $u_1,\cdots,u_d$ to be the $d$ rows of the matrix $DV^T$ where $UDV^T=A$ is the SVD of $A$.
It was proved in~\cite{FSS13} that taking the first $O(k/\eps)$ rows of $DV^T$ yields such a sketch, i.e., of size independent of both $n$ and $d$.
%Unlike the coreset of~\cite{kastiru-sensitivity}, this sketch is of cardinality independent of both $n$ and $d$.
Using the merge-and-reduced technique, this sketch can be computed in the streaming and parallel computation models. The final coreset size will be increased by a factor of $O(\log n)$ in this case.
For the streaming case, it was shown in~\cite{liberty2013simple, ghashami2015frequent} that such a strong sketch of the same size, $O(k/\eps)$, can be maintained without the additional $O(\log n)$ factor. The first sketch for sparse matrices was suggested in~\cite{clarkson2013low}, but like more recent results, it assumes that the complete matrix fits in memory.
%This result is based on a recent line of work that started in~\cite{}, and it is based on the item frequency estimation problem. In this paper we suggest a new variant of this problem which is one of our main technical tool; see Section... Lately~\cite{}, a heuristics that is based on the item frequency estimation problem, claims to compute strong sketches that is better than the previous approach in practice. A lot of other strong sketches were suggested over the years. See~\cite{ philp1,philip2 } for an excellent survey from this year.

\textbf{Lower bounds.} Recently,~\cite{ghashami2015frequent, liberty2013simple} proved a lower bound of $O(k/\eps)$ for the cardinality of a strong sketchs. A lower bound of $O(k/\eps^2)$ for strong coreset was proved for the special case $k=d-1$ in~\cite{batson2012twice}.
%
%\textbf{Reduced SVD for Sparse Data.} The classic SVD algorithm~\cite{} computes the reduced SVD exactly in $O(dn^2, nd^2)$ time. The state-of-the-art for computing approximation for the reduced SVD are based on computing this algorithm on the coreset or sketch. In~\cite{} Woodruf and Clarkson presented the first algorithm for computing an approximation to the low-rank approximation of a sparse matrix $A$, in time that is polynomial in the input sparsity and independent of the matrix size, $n$ and $d$. However, also after consulting with one of the authors, we couldn't find a way to compute low rank approximation for a matrix that is too large to fit in memory.
% %For example, since the $\tilde{A}$ is dense, unlike the input $A$, we cannot merge and re-compress it using the same algorithm.

\textbf{The Lanczoz Algorithm.}
The Lanczoz method and its variant as the Power Method are based on multiplying a large matrix by a vector for a few iterations to get its largest eigenvector $v_1$. Then the computation is done recursively after projecting the matrix on the hyperplane that is orthogonal to $v_1$. Multiplying a matrix by a vector can be done easily in the streaming mode without having all the matrix in memory.
%If $A$ is sparse then the computation of $Ax_i$ for $1<i<m$ takes time that depends only on the non-zeroes-entries of $A$.
%Ignoring numerical issues and the fact that we need to run a lot of passes over the matrix,
The main problem with running the Lanczoz method on large sparse data is that it is efficient only for the case $k=1$: the largest eigenvector $v_1$ of a sparse matrix $A$ is in general not sparse even $A$ is sparse. Hence, when we project $A$ on the orthogonal subspace to $v_1$, the resulting matrix is dense for the rest of the computations $(k>1)$.
Indeed, our experimental results show that the sparse SVD function in MATLAB which uses this method runs faster than the exact SVD, but crashes on large input, even for small $k$.

%% file: meancoresets.tex
\section{Novel Coreset Constructions}
\label{novel}
Our main technical result is a new approach that yields coresets of size significantly lower than the state-of-the-art.
While this paper is focused on coresets for computing the reduced SVD, our approach implies a technique for constructing coresets that may improve most of the previous coreset constructions.

Recall that coresets as defined in this paper must approximate every query for a given set of queries. In this paper the queries are $k$-subspaces in $\REAL^d$. Unlike sketches, coresets must also be weighted subsets of the input.
In this sense, all existing coresets constructions that we aware of can be described as computing sensitivity (or leverage score, or importance) for each point, as described earlier, and then computing what is known as an $\eps$-sample; see~\cite{FL10}. For a given set $A$ of size $n$, and a set of queries (subsets) of $A$, a subset $S\subseteq A$ is an $\eps$-sample if for every query, the fraction $f$ of $A$ that is covered by the query is $\eps$-approximated by the fraction $\tilde{f}$ that it covers from $S$, i.e,
\begin{equation}\label{items}
|f-\tilde{f}|\leq \eps .
\end{equation}
Such $\eps$-sample can usually be constructed deterministically in time that is exponential in the pseudo-dimension $d$ of $A$ (a complexity measure of shapes that is similar to VC-dimension~\cite{FL10}) , or using uniform random sampling of size the is linear $d$. The second approach yields randomized constructions which is known to be sub-optimal compared to deterministic ones~\cite{batson2012twice}.

In this section we suggest a new gradient descent type of deterministic construction of coresets by reducing coreset problems to what we call Item Frequency $\ell_2$-Approximation. The name comes from the well known Item Frequency Approximation (IFA) problem that was recently used in~\cite{liberty2013simple} to deterministically compute a sketch of size $O(k/\eps)$ for the reduced SVD problem. Unlike the $\eps$-sample and previous deterministic approach, or approach yields deterministic coresets of size independent of the pseudo-dimension $d$. This property will turn into coresets of size independent of $d$ for the reduced SVD problem and also for the simpler $1$-mean coreset that will be described below.

To show the relation to~\eqref{items}, we present the IFA problem a bit differently than in existing papers. In the IFA problem there is a universe of $n$ items $I=\br{e_1, \cdots , e_d}$ and a stream $a_1, \cdots, a_n\in I$ of item appearances. The frequency $f_i$ of an item $a_i$ stands for the fraction of times $e_i$ appears in the stream (i.e., its number of occurrences divided by $n$).
It is trivial to produce all item frequencies using $O(d)$ space simply by keeping
a counter for each item. The goal is to use $O(1/\eps)$ space and
produce approximate frequencies $\tilde{f}_1,\cdots,\tilde{f}_n$ such that
\[
 | f_i - \tilde{f}_i| \leq \eps.
\]
Note the surprising similarity to $\eps$-samples in~\eqref{items}.

As stated there~\cite{liberty2013simple, ghashami2015frequent}, IFA is the fundamental tool that was used to construct sketches of size $O(k/\eps)$. However, it was also proved in~\cite{ghashami2015frequent} that this approach \emph{cannot} be used to construct coresets (weighted subsets). This was suggested there as an open problem. In this paper we solve this problem by generalizing the IFA problem as follows.

\textbf{IFA equals $\ell_{\inf}$-IFA.}
First observe that the input vectors for the IFA problem can be considered as the rows of the $d\times d$ identity matrix $I$. We then wish to approximate each entry of the resulting mean vector by a small weighted subset,
\[
\max_i  | f_i - \tilde{f}_i| =\left\lVert\frac{\sum_i a_i}{n}-\sum_i w_i a_i\right\rVert_{\infty}\leq \eps,
\]
where $w\in[0,\infty)^n$ is a non-negative sparse vector of weights, and $\norm{(x_1,\cdots,x_n)}_\infty=\max_i |x_i|$ for $x\in\REAL^n$ is known as the $\ell_{\infty}$ norm.
More generally we can replace $I$ by all the unit vectors in $\REAL^d$.

\textbf{$\ell_2$-IFA.}
In this paper we define the $\ell_2$ version of IFA as the problem of computing a sparse non-negative vector  $w\in[0,\infty)^n$ such that
\[
\left\lVert\frac{\sum_i a_i}{n}-\sum_i w_i a_i\right\rVert_{2}\leq \eps,
\]
where $\norm{(x_1,\cdots,x_n)}_2=\sqrt{\sum_{i=1}^n x_i^2}$ is known as the $\ell_2$, Frobenius, or Euclidean norm.

Let $D^n$ denote the union over every vector $z\in [0,1]^n$ that represent a distribution, i.e., $\sum_i z_i=1$.
Our first technical result is that for any finite set of unit vectors $a_1,\cdots,a_n$ in $\REAL^d$, any distribution $z\in D^n$, and every $\eps\in(0,1]$, we can compute a distribution $w\in D^n$ of sparsity (non-zeroes entries) $\norm{w}_0\leq 1/\eps^2$. This result is in fact straight-forward by applying the Frank-Wolfe algorithm~\cite{clarkson2010coresets}, after normalizing the input points.
\begin{theorem}[\label{thm1}Item Frequency $\ell_2$ Approximation]
Let $z\in D^n$ be a distribution over $n$ vectors $a_1,\cdots,a_n$ in $\REAL^d$.
There is a distribution $w\in D^n$ of sparsity $\norm{w}_0\leq 1/\eps^2$ such that
\[
\norm{\sum_i z_i \cdot a_i-\sum_i w_i a_i}_{2}\leq \eps.
\]
\end{theorem}

%For the case that $q_i$ are standard unit vectors, we obtain a result that is strongly similar to the item frequency approximation problem; see~\cite{Edo Liberty, simpleMatrixSketching.pdf}
%
%The following theorem implies that every set has a small subset whose mean approximates the mean of the original set in term of sum of squared distances to the input points and absolute distance. For the simple case of (unweighted) sum of distances, we use 
%We can obtain Theorem 1 by applying Theorem~\ref{cor} with $y=\min_x \sum_i z_i\norm{p_i-x}^2$. However, the higher sparsity of $w$ also generalize the approximation to any given center $y\in\REAL^d$.
%The gradient-descent type of algorithm that computes $w$ is sketched in Fig.~\ref{k0} and explained in details in the supplementary material.
%We also show in the supplementary material that this theorem can be proved by applying the Frank-Wolfe algorithm~\cite{clarkson2010coresets}.
%Our self contained proof and faster construction is based on simple geometry and a solution to the problem of farthest point from a given point in a set. In the future, this might allow us to use more involved data structures or random projections for improving the running time of the construction. It can also be generalized to optimization problems on non-convex and non-differential functions, when it is less clear how to apply the Frank-Wolfe algorithm, similar to the ideas in the appendix of~\cite{ShyVar07}, in particular, to approximate the sum of distances to the power of $q$ for a given $k$-subspace.

\textbf{The Caratheodory Theorem} proves Theorem~\ref{thm1} for the special case $\eps=0$ using only $d+1$ points. Our approach and algorithms for $\ell_2$-IFA can thus be considered as $\eps$-approximations for the Caratheodory Theorem, to get coresets of size independent of $d$. Note that a Frank-Wolfe-style algorithm might run more than $d+1$ or $n$ iterations without getting zero error, since the same point may be selected in several iterations. Computing in each iteration the closest point to the origin that is \emph{spanned} by {all} the points selected in the previous iterations, would guarantee coresets of size at most $d+1$, and fewer iterations. Of course, the computation time of each iteration will be much slower in this case. Due to lack of space we omit further details.

%%%%%%%%%%%%%%%%%%%%%%%%%%%%%%%%%%%%%%%%%%%%%%%%%%%%%%%%%%%%%%%%
%\begin{figure}[t]
%\centering
%\subfloat[Iteration 1]{\includegraphics[trim=0 150 0 150,clip,width=0.33\columnwidth]{k01}\label{fig:k01}}
%\subfloat[Iteration 2]{\includegraphics[trim=0 150 0 150,clip,width=0.33\columnwidth]{k02}\label{fig:k02}}
%\subfloat[Iteration 3]{\includegraphics[trim=0 150 0 150,clip,width=0.33\columnwidth]{k03}\label{fig:k03}}
%\caption{Given $n=4$ points $a_1,\cdots, a_n$ in $\REAL^d=\REAL^2$ whose weighted mean $\sum_i z_i a_i$ is the origin, we let $c_1$ be an arbitrary input point. Then we iteratively find the farthest input point $a$ from $c_t$, and define $c_{t+1}$ to be the projection of the origin on the segment between $a$ and $c_t$ for $t=1,\cdots,N$ and $N=1/\eps^2$. We output the weight vector $w\in D^n$ that satisfies $c_N=\sum_{i=1}^n w_i a_i$.}
%\label{k0}
%\hrule
%\end{figure}
%%%%%%%%%%%%%%%%%%%%%%%%%%%%%%%%%%%%%%%%%%%%%%%%%%%%%%%%%%%%%%%%

\subsection{Warm up: reduction to $1$-mean $\eps$-coresets.}
%Note that $\sum_i z_i a_i$ is the mean over the distribution $z$ on the input vectors.
Given a set $a_1,\cdots,a_n$ of $n$ points in $\REAL^d$, and $\eps\in(0,1)$, we define a \emph{$1$-mean $\eps$-coreset} to be a weight vector $w\in[0,\infty)^n$ such that the sum of squared distances $\sum_{i=1}^n \norm{a_i-c}^2$ for every center $c\in\REAL^d$ to these input points is the same as the weighted sum $\sum_{i=1}^n w_i\norm{a_i-c}^2$ of the weighted points, up to a $(1\pm\eps)$ multiplicative factor.

Using the sensitivity framework~\cite{FL10}, we can compute such a coreset that has size $O(d/\eps^2)$ with high probability~\cite{LS10}. A corollary to this also answers the open problem in~\cite{clarkson2010coresets}.
Although the approximation property should hold for every center $c$, we reduce the $1$-mean coreset problem to the $\ell_2$ IFA problem, whose solution depends only on the input points. By scaling and normalizing the input points, we prove that Theorem~\ref{thm1} can be used to deterministically compute the first coreset of size independent of $d$ for the $1$-mean problem in time that is not exponential in $1/\eps$.

\begin{corollary}[$1$-mean $\eps$-coreset]
For every set of of $n$ vectors in $\REAL^d$ and every $\eps\in(0,1)$ a $1$-mean $\eps$-coreset of size $O(1/\eps^2)$ can be computed deterministically in $O(nd/\eps)$ time.
\end{corollary}

%% file: itemfreq.tex
\section{Reduction from Coreset to $\ell_2$-Item Frequency Approximation}
In the previous section we suggested the problem of $\ell_2$-IFA and several algorithms to solve it.
We then suggested a new construction for $1$-mean coreset that is based on a simple reduction to the $\ell_2$-IFA problem using a scaled versions of the input points.
In this section we prove a more involved reduction from the main problem of this paper: computing a $(k,\eps)$-coreset for the reduced SVD problem of a matrix $A$. In this case the input to the $\ell_2$-IFA problem is a set of $n$ $d\times d$  matrices, i.e., vectors in $\REAL^{d\times d}$. Each such matrix has the form $v_iv_i^T$ where $v_i$ is related to the $i$th row of the $U$ matrix in the Singular Value Decomposition $UDV^T$ of $A$. The proof of Theorem~\ref{thm:one} then follows by bounding the right hand side of~\eqref{rrr} using Theorem~\ref{thm1}. This is done by normalizing the input set, where the mean is translated and scaled to be the vector $(1,\cdots,0)$. To get an $\eps$-additive error after this scaling, the number of iterations is at most the sum of the norms of the matrices $v_iv_i^T$ divided by $\eps^2$, which yields a coreset of size $k/\eps^2$. 

\textbf{Notation.}
We denote by $\REAL^{n\times d}$ the set of all $n\times d$ matrices. For a given integer $i\geq 1$ we denote $[n]=\br{1,\cdots,n}$. For a pair of indexes (integers) $i\in[n]$, $j\in[d]$ and a given matrix $X\in\REAL^{n\times d}$  we denote by $X_{i,j}$ its entry in the $i$th row and $j$th column. As in Matlab, we denote by $j:k$ the set of indexes $\br{j,j+1,\cdots,k}$ for an integer $k\geq j$.
The $i$th row of $X$ is denoted by $X_{i,:}\in\REAL^{1\times d}$ and the $j$th column by $X_{:,j}\in\REAL^{n\times 1}$.
The $i$th entry of a vector $x=(x_1,\cdots,x_d)$ is denoted by $x_i$.

The Frobenius norm (root of squared entries) of a matrix or a vector $X$ is denoted by
$\norm{X}_F=\norm{X}=\sqrt{\sum_{i=1}^n\sum_{j=1}^d X_{i,j}^2}$. The
identity matrix is denoted by $I$ and the
matrix whose entries are all zeroes is denoted by $\mathbf{0}$. The size of $I$ and $\mathbf{0}$ is determined by the context. %The set $S^n\subseteq \REAL^{n\times 1}$ is the union of all vectors whose entries are non-negative and sum to one.
%The sparsity (number of non-zeros entries) of $X$ is denoted by $\norm{X}_0=|\br{x_{i,j} \mid i\in[n],j\in[d]} | $.
The inner product of a pair of matrices $X, Y\in\REAL^{n\times d}$ is denoted by $X\bullet Y=\sum_{i=1}^n\sum_{j=1}^d X_{i,j}Y_{i,j}.$
%The $i$th standard vector in the standard base of $\REAL^d$ is denoted by $e_i=I_{:,i}$.

We are now ready to prove the reduction from coresets to the Item Frequency $\ell_2$ Approximation problem, as follows.
Theorem~\ref{thm:one} then follows by using 

\begin{claim}\label{claima}
Let $A\in\REAL^{n\times d}$ be a matrix of rank $d$, and let $U\Sigma V^T=A$ denote its full SVD. % and let $D=\Sigma_{k+1:d,k+1:d}$.
%, and $\sigma=(\sigma_1,\cdots,\sigma_d)$ denote the diagonal of $\Sigma$ in a non-increasing order.
%whose singular values are $\sigma_1\geq \sigma_2\geq \sigma_d>0$.
Let $k\in[1,d-1]$ be an integer, and for every $i\in[n]$ let
\begin{equation}\label{rrr}
v_i=\left(U_{i,1},\cdots,U_{i,k},\frac{U_{i,k+1:d}\Sigma_{k+1:d,k+1:d}}{\norm{\Sigma_{k+1:d,k+1:d}}},
1\right).
%\left(\frac{A_{i,1}}{\sigma_1}, \cdots,
%\frac{A_{i,k}}{\sigma_{k}},\frac{A_{i,k+1:d}}{\sqrt{\sum_{j=k+1}^d\sigma^2_j}}\right)^T,
\end{equation}
Let $W\in\REAL^{n\times n}$ be a diagonal matrix.
Then for every $X\in\REAL^{d\times (d-k)}$ such that $X^TX=I$ we have
\[
\left|1-\frac{\norm{WAX}^2}{\norm{AX}^2}\right|
\leq
5\cdot \norm{\sum_{i=1}^n (1-W_{i,i})v_iv_i^T}.
\]
\end{claim}
\begin{proof}
%Let $\eps=\norm{\sum_i t_i(v_iv_i^T)-I\sum_i t_i(v_iv_i^T)}/\sqrt{2}$.
Let $\eps=\norm{\sum_{i=1}^n (1-W^2_{i,i})v_iv_i^T}.$
For every $i\in[n]$ let $t_i=1-W^2_{i,i}$.
Put $X\in\REAL^{d\times (d-k)}$ such that $X^TX=I$. Without loss of generality we assume $V^T=I$, i.e., $A=U\Sigma$, otherwise we replace $X$ by $V^TX$.
It thus suffices to prove that
\begin{equation}\label{toprov}
|\sum_i t_i \norm{A_{i,:}X}^2|\leq 5\eps\norm{AX}^2.
\end{equation}
Using the triangle inequality,
\begin{align}
\nonumber&|\sum_i t_i\norm{A_{i,:}X}^2|\\
\label{bb}&\leq
\left|\sum_i t_i\norm{A_{i,:}X}^2-\sum_i t_i\norm{(A_{i,1:k},\mathbf{0})X}^2\right|.
\\
\label{aa}&\quad+|\sum_i t_i\norm{(A_{i,1:k},\mathbf{0})X}^2|.
\end{align}
We now bound the last two expressions.

\textbf{Bound on~\eqref{aa}:}
It was proven in~\cite{felphd} that for every pair of $k$-subspaces $S_1,S_2$ in $\REAL^d$ there is $u\geq 0$ and a $(k-1)$-subspace $T\subseteq S_1$ such that the distance from every point $p \in S_1$ to $S_2$ equals to its distance to $T$ multiplied by $u$. By letting $S_1$ denote the $k$-subspace that is spanned by the first $k$ standard vectors of $\REAL^d$, letting $S_2$ denote the $k$-subspace that is orthogonal to each column of $X$, and $y\in\REAL^{k}$ be a unit vector that is orthogonal to $T$, we obtain that for every row vector $p\in\REAL^k$,
\begin{equation}\label{AA}
\norm{(p,\mathbf{0})X}^2=u^2(py)^2.
\end{equation}

After defining  $x=\Sigma_{1:k,1: k} y/\norm{\Sigma_{1:k,1: k} y}$,~\eqref{aa} is bounded by
\begin{align}
&\nonumber\sum_i t_i\norm{(A_{i,1:k},\mathbf{0})X}^2
=\sum_i t_i\cdot u^2\norm{A_{i,1:k}y}^2\\
\nonumber&=u^2\sum_it_i\norm{A_{i,1:k}y}^2\\
\nonumber&=u^2\sum_it_i\norm{U_{i,1:k}\Sigma_{1:k,1: k} y}^2\\
\label{efg}&=u^2\norm{\Sigma_{1:k,1: k} y}^2\sum_it_i\norm{(U_{i,1:k})x}^2.
\end{align}

\vspace{-0.3cm}
The left side of~\eqref{efg} is bounded by substituting $p=\Sigma_{j,1:k}$ in~\eqref{AA} for $j\in[k]$, as
%\vspace{-0.3cm}
\begin{align}
&\nonumber u^2\norm{\Sigma_{1:k,1:k}y}^2=\sum_{j=1}^k u^2(\Sigma_{j,1:k}y)^2
=\sum_{j=1}^k \norm{(\Sigma_{j,1:k},\mathbf{0})X}^2\\
\nonumber&=\sum_{j=1}^k \sigma^2_j \norm{X_{j,:}}^2\leq \sum_{j=1}^d \sigma^2_d \norm{X_{j,:}}^2\\
&= \norm{\Sigma X}^2=\norm{U\Sigma X}^2
\label{ggg}= \norm{A X}^2.
\end{align}
The right hand side of~\eqref{efg} is bounded by
\begin{align}
&\nonumber|\sum_it_i\norm{(U_{i,1:k})x}^2|
=|\sum_i t_i (U_{i,1:k})^TU_{i,1:k}\bullet xx^T|\\
\nonumber&=|xx^T\bullet \sum_i t_i (U_{i,1:k})^TU_{i,1:k}|\\
&\label{abc}\leq \norm{xx^T}\cdot\norm{\sum_i t_i (U_{i,1:k})^TU_{i,1:k}}\\
\label{efff}&\leq \norm{\sum_i t_i (v_{i,1:k})^Tv_{i,1:k}}\leq \norm{\sum_i t_i v_{i}^Tv_{i}}
=\eps,
\end{align}
where~\eqref{abc} is by the Cauchy-Schwartz inequality and the fact that $\norm{xx^T}=\norm{x}^2=1$, and in~\eqref{efff} we used the assumption $A_{i,j}=U_{i,j}\sigma_j=v_{i,j}$ for every $j\in[k]$.

Plugging~\eqref{ggg} and~\eqref{efff} in~\eqref{efg} bounds~\eqref{aa} as
\begin{equation}
\label{ef}
|\sum_i t_i\norm{(A_{i,1:k},\mathbf{0})X}^2|
\leq \eps\norm{AX}^2.
\end{equation}

\vspace{-0.2cm}
\textbf{Bound on~\eqref{bb}:}
For every $i\in[n]$ we have
\begin{align}
&\nonumber\norm{A_{i,:}X}^2-\norm{(A_{i,1:k},\mathbf{0})X}^2\\
&\nonumber= 2(A_{i,1:k},\mathbf{0})XX^T(\mathbf{0},A_{i,k+1:d})^T\\
\nonumber&\quad+\norm{(\mathbf{0},A_{i,k+1:d})X}^2\\
\nonumber&= 2A_{i,1:k}X_{1:k,:}(X_{k+1:d,:})^T(A_{i,k+1:d})^T\\
\nonumber&\quad+\norm{(\mathbf{0}, A_{i,k+1:d})X}^2\\
\end{align}
\begin{align}
\nonumber&= 2\sum_{j=1}^k A_{i,j}X_{j,:}(X_{k+1:d,:})^T(A_{i,k+1:d})^T\\
\nonumber&\quad+\norm{(\mathbf{0}, A_{i,k+1:d})X}^2\\
\nonumber&=  \sum_{j=1}^k 2\sigma_jX_{j,:}(X_{k+1:d,:})^T\cdot  \norm{\sigma_{k+1:d}}v_{i,j}(v_{i,k+1:d})^T\\
\nonumber&\quad+\norm{\sigma_{k+1:d}}^2\norm{(\mathbf{0},v_{i,k+1:d})X}^2.
%\label{fff}&\leq \sum_{j=1}^k 2\sigma_j\norm{X_{j,:}}\cdot \norm{\sigma_{k+1:d}}v_{i,j}(v_{i,k+1:d})^T\\
%\label{ggg}&\quad+\norm{\sigma_{k+1:d}}^2\norm{v_{i,k+1:d}}^2.
\end{align}

\vspace{-0.3cm}
Summing this over $i\in[n]$ with multiplicative weight $t_i$
and using the triangle inequality, will bound~\eqref{bb} by
\begin{align}
\nonumber&\left|\sum_i t_i\norm{A_{i,:}X}^2-\sum_i t_i\norm{(A_{i,1:k},\mathbf{0})X}^2\right|\\
&\label{fff}\leq \Big|\sum_i t_i \sum_{j=1}^k 2\sigma_jX_{j,:}(X_{k+1:d,:})^T\\
\nonumber&\quad\cdot  \norm{\sigma_{k+1:d}}v_{i,j}(v_{i,k+1:d})^T\Big|\\
\label{ggh}&\quad+\left|\sum_i t_i\norm{\sigma_{k+1:d}}^2\norm{(\mathbf{0},v_{i,k+1:d})X}^2\right|.
\end{align}
The right hand side of~\eqref{fff} is bounded by
\begin{align}
\nonumber&\left|\sum_{j=1}^k 2\sigma_jX_{j,:}(X_{k+1:d})^T\cdot
\norm{\sigma_{k+1:d}} \sum_i t_iv_{i,j}(v_{i,k+1:d})^T\right|\\
\label{ab1}&\leq  \sum_{j=1}^k 2\sigma_j\norm{X_{j,:}X_{k+1:d}}\cdot \norm{\sigma_{k+1:d}}\norm{\sum_i t_iv_{i,j}v_{i,k+1:d}}\\
%\nonumber&= \sum_{j=1}^k 2\sqrt{\eps}\sigma_j\norm{X_{j,:}} \cdot\frac{\norm{\sigma_{k+1:d}}}{\sqrt{\eps}}
%\norm{\sum_i t_i v_{i,j}v_{i,k+1:d}}\\
\label{ab2}&\leq \sum_{j=1}^k  (\eps\sigma^2_j\norm{X_{j,:}}^2
 +\frac{\norm{\sigma_{k+1:d}}^2}{\eps} \norm{\sum_i t_i v_{i,j}v_{i,k+1:d}}^2)\\
%\nonumber&=\eps\sum_{j=1}^k  \sigma^2_j\norm{X_{j,:}}^2
% +\frac{\norm{\sigma_{k+1:d}}^2}{\eps} \norm{\sum_i t_i (v_{i,1:k})^Tv_{i,k+1:d}}^2\\
\label{ee}&\leq 2\eps\norm{AX}^2,
\end{align}
where~\eqref{ab1} is by the Cauchy-Schwartz inequality, ~\eqref{ab2} is by the inequality $2ab\leq a^2+b^2$. In~\eqref{ee} we used the fact that $\sum_i t_i (v_{i,1:k})^Tv_{i,k+1:d}$ is a block in the matrix $\sum_i t_iv_iv_i^T$, and
\begin{equation}\label{opt}
\begin{split}
&\norm{\sigma_{k+1:d}}^2\leq \norm{AX}^2 \text{\quad and } \sum_{j=1}^k\sigma^2_j\norm{X_{j,:}}^2\\
&=\norm{\Sigma_{1:k,1:k}X_{1:k,:}}^2\leq \norm{\Sigma X}^2\leq \norm{AX}^2.
\end{split}
\end{equation}
%\vspace*{-0.05cm}
%Summing
%Plugging~\eqref{ee} in~\eqref{fff} and summing~\eqref{ggg} bounds~\eqref{bb} as
%\[
%\begin{split}
%|\sum_i t_i \left(\norm{A_{i,:}X}^2-\norm{(A_{i,1:k},\mathbf{0})X}^2\right)|
%&\leq 2\eps\norm{AX}^2+\norm{\sigma_{k+1:d}}^2|\sum_i t_i\norm{v_{i,k+1:d}X_{k+1:d,:}}^2|\\
%&\leq 2\eps\norm{AX}^2+\norm{AX}^2\cdot |\sum_i t_iv^2_{i,d+1}|\leq 3\eps\norm{AX}^2,
%\end{split}
%\]
%where in the second inequality we used the left side of~\eqref{opt} and the fact that the rows of $X$ are the leftmost $d-k$ columns of some orthogonal matrix.
Next, we bound ~\eqref{ggh}.
Let $Y\in \REAL^{d\times k}$ such that $Y^TY=I$ and $Y^TX=\mathbf{0}$. Hence, the columns of $Y$ span the $k$-subspace that is orthogonal to each of the $(d-k)$ columns of $X$.
By using the Pythagorean Theorem and then the triangle inequality,
\begin{align}
&\norm{\sigma_{k+1:d}}^2\label{aab} |\sum_i t_i \norm{(\mathbf{0},v_{i,k+1:d})X}^2|\\
=&\nonumber \norm{\sigma_{k+1:d}}^2|\sum_i t_i \norm{(\mathbf{0},v_{i,k+1:d})}^2\\
\nonumber &\quad-
\sum_i t_i \norm{(\mathbf{0},v_{i,k+1:d})Y}^2|
\\
&\leq \label{c1}\norm{\sigma_{k+1:d}}^2|\sum_i t_i \norm{v_{i,k+1:d}}^2|\\
&\quad+
\label{c2}\norm{\sigma_{k+1:d}}^2|\sum_i t_i \norm{(\mathbf{0},v_{i,k+1:d})Y}^2|.
%&\norm{\sigma_{k+1:d}}^2\label{c1}\leq |\sum_i t_i \norm{(v_{i,k+1:d},0)}^2|\\
%&\norm{\sigma_{k+1:d}}^2\label{c2}+|\sum_i t_i \norm{(v_{i,k+1:d},0)Y}^2|.
\end{align}
For bounding~\eqref{c2}, observe that $Y$ corresponds to a $(d-k)$ subspace, and $(\mathbf{0},v_{i,k+1:d})$ is contained in the $(d-k)$ subspace that is spanned by the last $(d-k)$ standard vectors. Using same observations as above~\eqref{AA}, there is a unit vector $y\in\REAL^{d-k}$ such that for every $i\in[n]$
$
\norm{(\mathbf{0},v_{i,k+1:d})Y}^2
=\norm{(v_{i,k+1:d})y}^2
$. Summing this over $t_i$ yields,
\[
\begin{split}
&|\sum_i t_i\norm{(\mathbf{0},v_{i,k+1:d})Y}^2|
=|\sum_i t_i\norm{v_{i,k+1:d}y}^2|\\
&=|\sum_i t_i \sum_{j=k+1}^d  v^2_{i,j}y_{j-k}^2|
\label{cd}=|\sum_{j=k+1}^d y_{j-k}^2 \sum_i t_i  v^2_{i,j}|.
\end{split}
\]
Replacing~\eqref{c2} in~\eqref{aab} by the last inequality yields

\input{algorithm}

\begin{align}
&\nonumber\norm{\sigma_{k+1:d}}^2|\sum_i t_i \norm{(\mathbf{0},v_{i,k+1:d})X}^2|\\
%&\nonumber\leq
%\norm{\sigma_{k+1:d}}^2(|\sum_i t_i \norm{v_{i,k+1:d}}^2|
%+\sum_{j=k+1}^d y_{j-k}^2 |\sum_i t_i  v^2_{i,j}|)\\
\label{aacd}&\leq
\norm{\sigma_{k+1:d}}^2 (|\sum_i t_i v^2_{i,d+1}|
+\sum_{j=k+1}^d y_{j-k}^2 \norm{\sum_i t_i  v_{i}v_i^T})\\
&\leq \norm{\sigma_{k+1:d}}^2(\eps
+\eps\sum_{j=k+1}^d y_{j-k}^2 )\label{ax}\leq 2\eps\norm{AX}^2,
\end{align}
where~\eqref{aacd} follows since $\sum_i t_i v_{i,j}^2$ is an entry in the matrix $\sum_i t_i v_iv_i^T$,
in~\eqref{ax} we used~\eqref{opt} and the fact that $\norm{y}^2=1$.
Plugging~\eqref{ee} in~\eqref{fff} and~\eqref{ax} in\eqref{ggg} gives the desired bound on~\eqref{bb} as
\[
|\sum_i t_i \norm{A_{i,:}X}^2-\sum_i t_i\norm{(A_{i,1:k},\mathbf{0})X}^2|
\leq  4\eps\norm{AX}^2.
\]
Finally, using~\eqref{ef} in~\eqref{aa} and the last inequality in~\eqref{bb}, proves the desired bound of~\eqref{toprov}.
\end{proof}

\subsection{ Implementation.} We now give a brief overview that bridges the gap between the theoretical results and the practical implementation presented in Algorithm 1.
The coreset construction described in this section use storage of $O(d^2)$ per point. However the suggested implementation uses $O(d)$ space for point. To achieve this we leverage several observations: (1) we do not need to compute the center, only its norm, which can we can update recursively from a starting value of 1; (2) the term $\sum_i w_i (x_i x_j^T)^2$ only needs to be computed $O(k/\eps^2)$ times, and its value stored for further iterations that find the same farthest point.

The algorithm works as follows. First we restructure the input points $A$ using their $k$-rank SVD decomposition $UDV^T$ into a new input matrix $P$, and normalize it to get $X$. We arbitrarily select starting point $X_j = X_1$ and set its weight to $w_j = 1$, with all other weights set to zero. Next, we compute the farthest point from $X_j$ by projecting the weighted points onto the current point. The following expressions $a,b,c$ contain all the update steps in order to compute the norm of the new center recursively without computing the actual center itself (an $O(d^2)$ computation). Finally, the value of $\alpha$ is updated based on the ratio of distances from the current point, current center, and new center, and the weights are updated based on the new $\alpha$. The algorithm runs for $k/\eps^2$ iterations, or until $\alpha$ has converged to 1.
 The output of the algorithm is a sparse vector of weights satisfying the guarantees.

%% file: algorithm.tex
\begin{algorithm}[t]
   \caption{$\SVDCoreset(A,k,\eps)$}
   \label{alg:1}
\begin{algorithmic}
\STATE {\bfseries Input:} $A$: $n$ input points $A_{1 \ldots n}$ in $\REAL^d$
\STATE {\bfseries Input:} $k$: the approximation rank
\STATE {\bfseries Input:} $\eps$: the nominal approximation error
\STATE {\bfseries Output:} $w\in[0,\infty)^n$: non-negative weights vector.

\STATE Compute $UDV^T=A$, the SVD of $A$
%\STATE $\sigma = \sum_{i=k+1}^d \rbr{D_{i,i}}^2$

%\STATE $U, X \gets 0^{n \times d} $
\STATE $R \gets D_{k+1:d,k+1: d}$
%\STATE $P \gets \set{p \mid p_i = (u_1,\cdots,u_k, u_{k+1}/\sigma,\cdots,u_d/\sigma)^T}$
%\STATE $R = U_{1\ldots n,k+1\ldots d}$
\STATE $P \gets$ matrix with $i\in[n]$ row is $(U_{i,1:k},U_{i,k+1: d} \cdot \frac{R}{\norm{R}_F})$
\STATE $X \gets$ matrix with $i\in[n]$ row is $X_i = P_i/\norm{P_i}$
\STATE $w \gets (1,0,\cdots,0)$

\FOR{$i = 1,\ldots,{k}/{\eps^2}$}
\STATE $j \gets \mathrm{argmin}_{i}\set{{wX X_i}}$
\STATE $a = \sum_{i=1}^n w_i (X_i^T X_j)^2$
\STATE $b =  \rbr{1 - \norm{PX_j}_F^2 + \sum_{i=1}^n w_i \norm{PX_i}_F^2} / {\norm{P}_F^2}$
\STATE $c = \norm{wX}_F^2$
%\STATE $g = $
%\STATE $h = $
\STATE $\alpha = \rbr{1 - a + b}/\rbr{1 + c - 2a}$
\STATE $w \gets (1-\alpha)I_{j,:} + \alpha w$

\ENDFOR
\STATE {\bfseries return} $w$

\end{algorithmic}
\vspace{-0.1cm}
\end{algorithm}
\vspace{-0.3cm}

%% file: conclusion.tex
\section{Conclusion}
\label{sec: conclusion}

We present a new approach for dimensionality reduction using coresets.
%Our solution is general and can be used to project spaces of dimension $d$ to subspaces of dimension $k < d$.
The key
feature of our algorithm is that it computes coresets that are small in size and subsets of the original data.
%We use our coreset to compute the SVD of very large datasets.
Using synthetic data as ground truth we show that our algorithm provides a good approximation. 
%We then show that we can compute the SVD of the entire English Wikipedia in under 4 hours -- a computation task currently not possible with state of the art algorithms. We see this work as a theoretical foundation as well as a practical toolbox for a range of dimensionality reduction problems. We are currently developing instantiations for PCA, linear regression, LDA, and NNMF. 